\documentclass[conference]{IEEEtran}

\usepackage{cite}
\usepackage{amsmath,amssymb,amsfonts}
\usepackage{graphicx}
\usepackage{textcomp}

\def\BibTeX{{\rm B\kern-.05em{\sc i\kern-.025em b}\kern-.08em
    T\kern-.1667em\lower.7ex\hbox{E}\kern-.125emX}}

\usepackage{algorithm}
\usepackage{algpseudocode}
\usepackage{xcolor}
\usepackage{amsthm}
\usepackage{hyperref}
\usepackage{graphics}
\usepackage{braket}
\newtheorem{theorem}{Theorem}

\newtheorem{definition}{Definition}

\newtheorem{assumption}{Assumption}
\newcommand{\etal}{\textit{et al. }}

\usepackage{umoline}

\usepackage{soul}

\DeclareMathOperator{\Tr}{Tr}
\begin{document}
\bstctlcite{IEEEexample:BSTcontrol}

\title{
	Guiding Agents of Quantum Games to Equilibrium using Matrix Exponential Fixed-Point Iteration
}

\author{
	\IEEEauthorblockN{
		\uppercase{Alireza~Habibi}\IEEEauthorrefmark{1}\IEEEauthorrefmark{2}, 
		\uppercase{Luis F. Abanto-Leon}\IEEEauthorrefmark{2}, 
		\uppercase{Setareh~Maghsudi}\IEEEauthorrefmark{2}\IEEEauthorrefmark{3}
	}
	\IEEEauthorblockA{
		\IEEEauthorrefmark{2}Faculty of Electrical Engineering and Information Technology, Ruhr University Bochum, Bochum, Germany, 
	}
	\IEEEauthorblockA{\IEEEauthorrefmark{3} Faculty of Computer Science, Ruhr University Bochum, Bochum, German
	} 
	\IEEEauthorblockA{\IEEEauthorrefmark{1}Corresponding author: Alireza~Habibi (email: alireza.habibi@ruhr-uni-bochum.de)}
	
}

\maketitle

\begin{abstract}
In recent years, quantum game theory has gained significant attention as a framework for studying decision-making in multi-agent systems using quantum principles.
However, computing equilibrium strategies is challenging because the dimension of the joint Hilbert space grows as the product of the players' local dimensions.
In this paper, we consider an extended Gutoski-Watrous (EGW) game in which each player's quantum strategy is represented by a local density matrix.
We derive tensor-contraction expressions for the payoff functions and their gradients, thereby avoiding the explicit construction of the full joint density matrix and its computationally expensive multiplication by the payoff operators.
Building on the resulting effective Hamiltonians, we propose the \textit{Matrix Exponential Fixed-Point Iteration with Annealing (MEFPIA)} algorithm to search for equilibrium points in EGW games.
We compare MEFPIA with the Matrix Multiplicative Weights Update (MMWU) algorithm in terms of convergence.
For the tested instances and parameter settings, both algorithms approach the same strategy profiles and payoffs, while MEFPIA achieves lower relative error in fewer iterations.
These results indicate that MEFPIA is a promising numerical method for equilibrium search in multi-agent quantum games.
Our findings provide important insights into the quantum game theory’s potential for addressing complex decision-making processes, as well as opening up new paths for future research and exploration in multi-agent quantum systems.
\end{abstract}

\section{Introduction}
\label{sec:introduction}
Finding
the optimal decisions and determining how to reach them play an important role in different fields like economics, political science, management, and artificial intelligence. 
The decision-making process includes analyzing the available information, considering possible outcomes, evaluating the preferences of the referees, selecting an optimal decision or strategy from a range of available options, and evaluating the optimal decision~\cite{edwards1954theory,camerer2004advances}.
In a pure strategy, the player selects only one action, whereas in a mixed strategy, the player assigns probabilities to multiple actions and chooses among them randomly.

Quantum game theory extends classical game theory by including quantum principles, such as superposition, interference, and entanglement, into the strategic interactions among multiple agents.
As early examples, Meyer used unitary operators as quantum strategies to create a superposition. The superposition allows the quantum system to be in multiple states at the same time until it is measured~\cite{meyer1999quantum}.
He showed that a player using quantum strategies can outperform an opponent restricted to classical strategies in games such as the quantum penny-flip game. 
Afterward, Eisert~\etal investigated the role of entanglement in quantum game theory~\cite{eisert1999quantum}.
They proposed the Eisert-Wilkens-Lewenstein (EWL) quantum Prisoner's Dilemma and showed that, in the maximally entangled setting, quantum strategies can produce a Nash equilibrium with a higher payoff than the classical mutual-defection equilibrium.
More recently, Koh~\etal studied a multiplayer quantum
volunteer's dilemma within the EWL framework and identified
symmetric Nash equilibria with higher expected payoffs than
the unique symmetric mixed-strategy equilibrium of the
corresponding classical game~\cite{koh2025volunteer}.
Hebdon and Koh subsequently investigated how these
equilibria depend on the degree of entanglement, showing
that maximal entanglement is not always
required~\cite{hebdon2026entanglement,agreda2026experimental}.
Gutoski and Watrous proposed a operator representation of quantum strategies in multi-round interactions~\cite{gutoski2007toward}.
In the noninteractive special case, each player submits a quantum state, represented by a density matrix, to a referee who determines the players' payoffs through a joint measurement~\cite{jain2009parallel}.
In this paper, we refer to the multiplayer version of this density-matrix formulation as the extended Gutoski-Watrous (EGW) game.
The concept related to this quantum game, especially those that use semidefinite structures, can have practical applications across different fields~\cite{huang2010semidefinite,so2010probabilistic,mertikopoulos2015learning}. For example, in signal processing, they can be used in multi-agent covariance matrix optimization that enable efficient distributed optimization for tasks like beamforming and interference alignment. Similarly, in multiple–input and multiple–output (MIMO) systems, semidefinite programming (SDP) relaxations are used to solve energy efficiency and resource allocation challenges to optimize performance under interference, uncertainty, and throughput constraints~\cite{mertikopoulos2015learning,mertikopoulos2017distributed}. These approaches can also find applications in network localization, control systems, and cognitive radio networks~\cite{huang2010semidefinite,yu2004iterative,wang2022cooperative,scutari2009competitive}.

The density-matrix formulation of quantum games connects quantum information with multi-agent optimization.
To find the optimal strategies in a quantum game theory, quantum learning becomes important. It helps players to modify their strategies based on the payoffs and interactions with other players. 
Bostanci and Watrous showed that computing an approximate Nash equilibrium in a broad class of quantum games is PPAD-complete, as is the corresponding problem for classical games~\cite{bostanci2022quantum}.
Tsuda~\etal developed the matrix exponential gradient method for learning in the context of matrix and dictionary learning~\cite{tsuda2005matrix}.
This algorithm is reformulated as Matrix Multiplicative Weights Update (MMWU) methods that generalize classical multiplicative-weights algorithms to matrix variables~\cite{arora2012multiplicative}. 
Despite its classical nature, MMWU has been successfully adapted to density-matrix strategy spaces and widely used for learning and equilibrium computation in multiplayer quantum games~\cite{lotidis2024payoff,jain2022matrix,aaronson2018online,lotidis2023stability,chen2024adaptive,lin2023quantum}.
For two-player zero-sum noninteractive quantum games, Jain and Watrous showed that the time-averaged MMWU strategies provide an approximate equilibrium~\cite{jain2009parallel}.
Jain~\etal subsequently studied a continuous-time analogue of MMWU, called quantum replicator dynamics~\cite{jain2022matrix}.
They showed that, in this system, the total quantum relative entropy is a constant of motion and the dynamics could fail to achieve equilibrium. 
Lotidis~\etal showed that pure quantum equilibrium can be stable and attracting under the MMWU dynamic ~\cite{lotidis2023stability}.

Regardless of the advantages of MMWU, its convergence rate is not sufficiently fast. This limitation comes from its dependence on the time-averaged density matrix, which is heavily influenced by the nature of the entropy used during the development. 
In addition, computing strategy updates in this setting requires repeated evaluations of the players' payoffs and their gradients, which can become computationally expensive and memory-intensive as the joint Hilbert space grows.
This motivates numerical methods that exploit the structure of the strategy profiles and can be implemented on classical computers.
This paper aims to address these issues.

\textbf{Our contributions:}
In this paper, we proposed a numerical approach to equilibrium search in EGW quantum games. Our main contributions are:
\begin{itemize}
	\item We derive tensor-contraction expressions for the payoff functions and their gradients. These expressions avoid explicitly constructing the joint density matrix and multiplying it by the reward operators, reducing computational and memory costs compared with direct evaluation using full matrix products.
	\item We propose the \textit{Matrix Exponential Fixed-Point Iteration with Annealing (MEFPIA)} algorithm, which combines matrix-exponential updates derived from entropy-regularized best responses with a gradually decreasing temperature. All learning and optimization steps of MEFPIA are carried out on a classical computer.
	\item We evaluate MEFPIA against MMWU in quantum games involving separable and entangled measurement operators, and investigate the effect of the cooling rate. For the tested instances and parameter settings, MEFPIA achieves lower relative payoff error in fewer iterations.
	\end{itemize}

The paper is structured as follows.
For the readers with a background in game theory and limited familiarity with quantum mechanics, section~\ref{sec:quantum-basics} provide the fundamental concepts of quantum mechanics used in this paper.
In Section~\ref{sec:notations}, we introduce the notation used in quantum game theory, including that for a composite quantum system in which each player controls a local subsystem.
Section~\ref{sec:model} defines the EGW quantum game and quantum strategies, then explains how the reward function is calculated. 
In Section~\ref{sec:theoratical-results}, we discuss how to compute the trace, partial trace, and gradient of the reward function. 
Section~\ref{sec:learning} covers the MEFPIA learning algorithm. 
Section~\ref{sec:expriments} presents the experimental evaluation of MEFPIA and MMWU, comparing their convergence speed and precision. 
At last, Section~\ref{sec:conclud} holds concluding remarks.

\section{Quantum mechanic basics}
\label{sec:quantum-basics}
Quantum mechanics is a fundamental theory in science that describes the behavior of matter and energy at the atomic and subatomic levels. Unlike deterministic classical mechanics, quantum mechanics is intrinsically probabilistic, which means that some system features cannot be exactly measured. Quantum mechanics also includes phenomena such as superposition, which allows particles to exist in multiple states simultaneously.
This framework is crucial for understanding modern technologies such as semiconductors, lasers, and quantum computers~\cite{nielsen2010quantum}.

In this paper, we use the Dirac notation, sometimes referred to as the bra-ket notation in quantum mechanics, as a powerful framework for representing vectors and operators in the Hilbert space ($\mathcal{H}$)~\cite{nielsen2010quantum}.
A pure quantum state is denoted either by the Greek letter $\psi$ or by a ket $\ket{\psi} \in \mathcal{H}$. The dual vector, or conjugate transpose, is represented by the bra $\bra{\psi} = (\ket{\psi})^\dag$. 
The inner product between two quantum states is written as $\braket{\phi | \psi}$. A quantum state is always normalized so that $\braket{\psi | \psi} = 1$. 
Sometimes, we are unable to represent the quantum state of a system using a single quantum state. In such cases, we can describe the probabilities of different quantum states using a density matrix, expressed as
\begin{align}
	\rho = \sum_i \alpha_i \ket{\psi_i}\bra{\psi_i},
\end{align}
where $\alpha_i$ is the probability of finding the system in the quantum state $\ket{\psi_i}$,with $\alpha_i\geq0$ and $\sum_i\alpha_i=1$.
The density matrix has the properties as
\begin{align}
	\label{eq:dm-conditions}
	\rho=\rho^\dagger,\quad
	\rho\succeq0,\quad
	\Tr(\rho)=1,\quad
	0<\Tr(\rho^2)\leq1.
\end{align}
The quantum state is considered pure if $\Tr(\rho^2) = 1$, which means that we can describe the state of the system using only one quantum state called pure quantum state; otherwise, it is a mixed state. The corresponding density matrix is then termed as a pure density matrix or a mixed density matrix, respectively.

Quantum mechanics uses operators ($\mathbf{O}: \mathcal{H} \to \mathcal{H}$) to express physical observables and state transformations. Hermitian operators ($\mathbf{O}^\dag = \mathbf{O}$) represent measurable quantities with real eigenvalues, whereas unitary operators ($\mathbf{U}^\dag \mathbf{U} = \mathbf{U} \mathbf{U}^\dag = \mathbf{I}$, where $\mathbf{I}$ is the identity operator. ) represent the evolution of the quantum state. In this paper, we use bold capital letters for operators. 
The expectation value of an observable $\mathbf{O}$ for the pure quantum state $\ket{\psi}$ is expressed as
\begin{align}
	\label{eq:epv}
	\braket{\mathbf{O}}_{\psi} = 
	\braket{ \psi | \mathbf{O} | \psi}.
\end{align}
This expression represents the statistical average of the measurement outcomes associated with the observable $\mathbf{O}$ when the system is in the state $\ket{\psi}$. 
For either a pure or mixed quantum state, the expectation value of an observable $\mathbf{O}$ for the density matrix $\rho$ is given by
\begin{align}
	\braket{\mathbf{O}}_{\rho} = \Tr(\mathbf{O}\rho).
\end{align}

\textbf{Quantum Measurement:}
In order to extract information from quantum states, it is necessary to have a set of measurement outcomes that can be used to allocate rewards to each player. In this paper, we use a positive operator-valued measure (POVM) to provide a general and flexible description of quantum measurements, particularly in scenarios that involve partial or nonorthogonal measurements.
Consider a finite set of measurement outcomes $\Omega$. To determine the probability of measuring a particular outcome $\omega \in \Omega$, we define a positive semidefinite operator ${{\mathbf{P}_\omega}: \mathcal{H} \to \mathcal{H}}$, subject to the condition ${\sum_{\omega}\mathbf{P}_{\omega} = \mathbf{I}}$. In this paper, we refer to $\mathbf{P}_\omega$ as the outcome operator corresponding to the measurement of the outcome $\omega$.
The probability of measuring the outcome $\omega$, given that the environment is in a pure quantum state $\ket{\psi}$, is 
\begin{align}
	P_{\omega}(\psi)=\Braket{\psi|\mathbf{P}_{\omega}|\psi}.
\end{align}
with $\sum_{\omega} P_{\omega}(\psi) =1$.
When the system is in a mixed state described by the density matrix $\rho$, the probability of measuring the outcome $\omega$ is given by $P_{\omega}(\rho) = \Tr(\mathbf{P}_{\omega} \rho)$. 
There are several methods to generate outcome operators~\cite{nielsen2010quantum,busch1996standard}. We explain a simple method using orthonormal vectors to generate POVM outcome operators in Appendix~\ref{append:povm_generate}.

\textbf{Von Neumann entropy:} 
This quantum entropy is defined as
\begin{align}
	S(\rho) = -\Tr (\rho \log \rho),
\end{align}
and measures the uncertainty or mixedness in a quantum system. For a quantum system with pure states (states without mixedness ), the Von Neumann entropy is zero, while maximally mixed states have the maximum possible entropy.

\textbf{Entangled and separable states:}
In quantum mechanics, entanglement arises when we have composite systems (of subsystems) that can not be separated. In this case, subsystems are correlated, meaning that the state of one subsystem cannot be described independently of the state of the others.
A separable system can be expressed as
\begin{align}
	\label{eq:separable-state-ent}
	\rho^{(\text{s})} = \sum_i p_i \, \rho^{(A)}_i \otimes \rho^{(B)}_i,
\end{align}
where $p_i$ are probabilities such that $\sum_i p_i = 1$ and 
the symbol $\otimes$ denotes the tensor product of quantum states.
Here, $\rho^{(A)}$ and $\rho^{(B)}$ are the density matrices of the subsystems $A$ and $B$, respectively. 
We use the superscript $(\text{s})$ in the tensor product space to insist on separable operators or density matrix.
Entanglement emerges when the density matrix $\rho$ cannot be decomposed in this way (Eq.~\eqref{eq:separable-state-ent}) which indicates a non-classical correlation between subsystems.

\textbf{Gibbs Distribution and Density Matrix:}  
Consider a quantum system at temperature $T$. When the total energy of the system is described by the Hamiltonian operator $\mathbf{H}$, the equilibrium state of the quantum system is represented by the Gibbs distribution. The Gibbs distribution is related to the density matrix as~\cite{trushechkin2022open}
\begin{align}
	\label{eq:gibs-density}
	\rho = \frac{e^{- \frac{\mathbf{H}}{k_B T}}}{\Tr\left(e^{-\frac{\mathbf{H}}{k_B T}}\right)},
\end{align}
where $k_B$ is Boltzmann's constant. In this paper, we set $k_B = 1$ for simplicity. In this formulation, the density matrix $\rho$ encodes the probabilities that the system is in different energy eigenstates.  
The Hamiltonian $\mathbf{H}$ determines the system's energy levels, while the parameter $T$ controls the distribution of these probabilities. When $T = 0$, the system settles into its ground state.
The ground state is the lowest-energy eigenstate. As $T \to \infty$, the density matrix becomes fully mixed, which reflects a uniform distribution over all states and energies.

\section{Notations}
\label{sec:notations}
In this paper, for a EGW quantum game with $N$ players where each player has access to their own subsystem, we use the following notation.
Consider a finite-dimensional complex Hilbert space $\mathcal{H} \cong \mathbb{C}^d$, where $d$ is its dimension\cite{nielsen2010quantum}. 
In a $N$-player quantum game, each player $i$ controls a local Hilbert space $\mathcal{H}^{(i)}$ of size $d_i$.  
The joint system is described by the tensor-product space
\begin{align}
	\mathcal{H} = \bigotimes_{i=1}^N \mathcal{H}^{(i)},
	\qquad
	d = \dim(\mathcal{H}) = \prod_{i=1}^N d_i.
\end{align}
Each player's strategy is represented by a density matrix
$\rho^{(i)} \in \mathbb{C}^{d_i \times d_i}$.
We use bold capital letters for other quantum operators.
An operator $\mathbf{O}$ acting on $\mathcal{H}$ has a $d \times d$ matrix representation in a chosen product basis.
To make its subsystem structure explicit, we replace the row and column indices by multi-indices,
$m \leftrightarrow (j_1,\ldots,j_N)$ and
$n \leftrightarrow (j'_1,\ldots,j'_N)$,
where $j_k,j'_k \in \{1,\ldots,d_k\}$.
The operator entries can then be written as
$O_{j_1,\ldots,j_N,j'_1,\ldots,j'_N}$,
giving an equivalent tensor representation with $2N$ indices.

We use the notation $(\mathbf{j}) \equiv (j_1, j_2, \cdots, j_N)$, where $j_i$ is the index $j$ corresponding to the player $i$. The notation $(\mathbf{j}_{-i})$ denotes $(j_1, j_2, \cdots, j_{i-1}, j_{i+1}, \cdots, j_N)$, where the index $j_i$ for the player $i$ is omitted. The $(j'_i; \mathbf{j}_{-i}) = (j_1, \cdots, j'_i, \cdots, j_N)$ represents that $j_i$ is replaced by $j'_i$.
The sum over $\mathbf{j}$ represents the sum over all indices as
\begin{align}
	\sum_{\mathbf{j}}= \sum_{ j_1,j_2,\cdots,j_N }= \sum_{j_1}\sum_{j_2}\ldots\sum_{j_N},
\end{align}
and the sum over $\mathbf{j}_{-i}$ denotes the sum over all indices except $j_i$
\begin{align}
	\sum_{\mathbf{j}_{-i}} = \sum_{j_1}\ldots\sum_{j_{i-1}}\sum_{j_{i+1}}\ldots\sum_{j_N}.
\end{align}

In addition, we distinguish local operators from operators acting on
the joint Hilbert space.
The operator $\mathbf{O}^{(i)}\in\mathbb{C}^{d_i\times d_i}$
acts on player $i$'s local Hilbert space $\mathcal{H}^{(i)}$,
whereas $\mathbf{O}_i\in\mathbb{C}^{d\times d}$ acts on
the joint Hilbert space $\mathcal{H}$ and is associated
with player $i$.
For a separable density matrix, we have
\begin{subequations}
	\begin{align}
		\rho^{(\text{s})} &= \rho^{(1)} \otimes \rho^{(2)} \otimes \cdots \otimes \rho^{(N)},
		\\
		\rho_i^{(\text{s})} &= \mathbf{I}^{(1)} \otimes \cdots \otimes \rho^{(i)} \otimes \cdots \otimes \mathbf{I}^{(N)},
		\label{eq:equantum-op-2}
		\\
		\rho_{-i}^{(\text{s})} &= \rho^{(1)} \otimes \cdots \otimes \mathbf{I}^{(i)} \otimes \cdots \otimes \rho^{(N)},
		\\
		(\rho^{'(i)} ;\rho_{-i}^{(\text{s})} ) &= (\rho^{(1)} \otimes \cdots \otimes \rho^{'(i)} \otimes \cdots \otimes \rho^{(N)}),
	\end{align}
\end{subequations}
where $\mathbf{I}^{(i)}$ is the identity matrix with the same dimension as the density matrix $\rho^{(i)}$.
For the complex conjugate of a complex number $a$, we use the notation $\bar{a}$.

A pure quantum state $\ket{\Psi}$ of an $N$-player system with $N$ subsystems can be expressed as
\begin{align}
	\ket{\Psi} = \sum_{\mathbf{j}} c_{\mathbf{j}} \ket{\mathbf{j}},
\end{align}
where $\mathbf{j}=(j_1,\ldots,j_N)$ is a composite index and
$\ket{\mathbf{j}}=\ket{j_1 j_2 \ldots j_N}
= \ket{j_1} \otimes \ket{j_2} \otimes \ldots \otimes \ket{j_N}$
is a product basis vector, with $j_i \in \{0,\ldots,d_i-1\}$.
The local basis vectors are orthonormal.
The complex probability amplitudes
$c_{\mathbf{j}}=c_{j_1,j_2,\ldots,j_N}$ form an order-$N$ tensor.
The normalization condition $\Braket{\Psi|\Psi}=1$ therefore requires
$\sum_{\mathbf{j}} |c_{\mathbf{j}}|^2=1$.
Similarly, for a system with $N$ subsystem, the general form of quantum operator is
\begin{align}
	\label{eq:nplayer-operator}
	\mathbf{A} =
	\sum\limits_{\mathbf{j},\mathbf{j}'}
	A_{\mathbf{j},\mathbf{j}'} 
	\ket{ \mathbf{j}}
	\bra{ \mathbf{j}'},
\end{align}
where $A$ denotes a Tensor rank $2N$.

\section{Model}
\label{sec:model}
In a EGW quantum game, players apply their strategies to their respective subsystems, and the reward is calculated based on the strategy profile and each player's reward operator.
A N-player EGW quantum game is a tuple $\mathcal{Q}_{\text{EGW}} = \langle \mathcal{N}, \mathcal{H}, \mathcal{S}, r \rangle$ defined as follows~\cite{gutoski2007toward,lotidis2024payoff} 
\begin{itemize}
	\item $\mathcal{N} = \{1, 2, \ldots, N\}$: A finite set of $N$ players.
	\item $\mathcal{H} = \mathcal{H}^{(1)} \otimes \mathcal{H}^{(2)} \otimes \ldots \otimes \mathcal{H}^{(N)}$: A Hilbert space associated with the game. Each player $i \in \mathcal{N}$ has access to a complex Hilbert space $\mathcal{H}_i \cong \mathbb{C}^{d_i}$, where $d_i$ represents the dimension of  player~$i$'s subsystem.

	\item $\mathcal{S}^{(i)}$: A mixed quantum strategy applied by each player $i\in \mathcal{N}$ is represented by a mixed density matrix ${\rho^{(i)} \in \mathbb{C}^{d_i} \times \mathbb{C}^{d_i}}$, satisfying the conditions mentioned in Eq.~\eqref{eq:dm-conditions}. A pure density matrix (${ \Tr({\rho^{(i)}}^2) = 1 }$) is a pure quantum strategy.  
	
	\item   $\mathcal{S}=\mathcal{S}^{(1)} \otimes \cdot \otimes \mathcal{S}^{(N)}$: The joint strategy space.
	A set of quantum strategies $\mathcal{S}$: 
	The strategy profile applied by all players is a separable density matrix expressed as
	\begin{align}
		\rho^{\text(\text{s})} = \rho^{(1)} \otimes \rho^{(2)} \otimes \ldots \otimes \rho^{(N)}.
	\end{align}
	
	\item The reward function $r$: Each player $i \in \mathcal{N}$ receives a real value reward, determined by an individual reward function $r_i: \rho \to \mathbb{R}$. The reward is determined by the outcomes of the POVM measurement process using a reward operator $\mathbf{R}_{i}: \mathcal{H} \to \mathcal{H}$. We assign a numerical outcome reward $R_i: \Omega \to \mathbb{R}$ to each outcome for player $i$, where $\Omega$ is a finite set of measurement outcomes. The reward operator can be represented as
	\begin{align}
		\label{eq:reward-op-main}
		\mathbf{R}_{i} = \sum_{\omega\in \Omega} R_i(\omega ) \mathbf{P}_\omega.
	\end{align}
	The expectation value of the reward for player $i$ is
	\begin{align}
		\label{eq:reward-function-density}
		r_i(\rho^{(\text{s})}) = \Braket{\mathbf{R}_{i}}_{\rho^{(\text{s})}} 
		= \Tr(\mathbf{R}_{i} \rho^{(\text{s})}).
	\end{align}
\end{itemize}

\textbf{Best Response (BR):}
Given the other players' strategy profile
$\rho^{(\text{s})}_{-i}$, the best-response set of player $i$ is
\begin{align}
	\operatorname{BR}_i\left(\rho^{(\text{s})}_{-i}\right)
	=
	\underset{\rho^{(i)}}{\arg\max}
	\; r_i\left(\rho^{(i)};\rho^{(\text{s})}_{-i}\right),
\end{align}
which maximize the reward of player $i$.

\textbf{Nash Equilibrium:}
A strategy profile $\rho^{(\text{s})*}=\bigotimes_{i=1}^{N}\rho^{(i)*}$
is a Nash equilibrium if no player can increase its payoff
by changing its own strategy while the other players'
strategies remain fixed,
\begin{align}
	r_i\left( \rho^{(i)*}; \rho^{(\text{s})*}_{-i}\right) \ge r_i\left( \rho^{(i)}; \rho^{(\text{s})*}_{-i}\right),\quad \forall i \in \mathcal{N}.
\end{align}
Equivalently, every player's equilibrium strategy is a
best response to the equilibrium strategies of the others,
\begin{align}
	\rho^{(i)*}
	\in \operatorname{BR}_i\left(\rho^{(\mathrm{s})*}_{-i}\right),
	\qquad \forall i \in \mathcal{N}.
\end{align}

Before proceeding further, we need to establish some foundational results related to the use of trace, multiplication, and derivatives in this context.
In the following section, we will outline these results.

\section{Theoretical Results}
\label{sec:theoratical-results}
The dimension of the density matrix, calculated using the Kronecker product, scales exponentially to $(\prod_i d_i )\times (\prod_i d_i) $. 
In order to calculate the reward function using Eq.~\eqref{eq:reward-function-density}, we need to multiply the quantum operator by the density matrix and calculate the trace. As a result, calculating the rewards and gradients for this task become computationally expensive and memory-intensive  with a high number of players and a large dimension of the strategy space. In the following, we simplify the task by applying the reward operator to the separable density matrix and calculating the trace using tensor contraction rules, which helps to reduce the computational cost and additional memory requirements compared with direct evaluation.
First, we begin by calculating the trace and partial trace of the quantum operator that consists of $N$ subsystems, as defined in Eq.~\eqref{eq:nplayer-operator}.
The trace of $\mathbf{A}$ over all players' subsystems and the partial trace of $\mathbf{A}$ over all players' subsystems except player $i$ are given by
\begin{subequations}
	\label{eq:trace_all}
	\begin{align}
		\Tr(\mathbf{A}) &= \sum_{\mathbf{j}} A_{\mathbf{j},\mathbf{j}},  \\
		\Tr_{-i}(\mathbf{A}) &=
		\sum_{j_i,j'_i} \Big( \sum_{\mathbf{j}_{-i}}
		A_{(j_i;\mathbf{j}_{-i}) , (j'_i;\mathbf{j}_{-i})} \Big) 
		\ket{j_i}\bra{j_i'},
	\end{align}
\end{subequations}
where $\Tr_{-i}$ denotes the partial trace over all indices except $i$.
%

% \re{lemma 1moved to appendix.... \\}
The partial trace helps us focus on a specific subsystem of a system while tracing out (or ignoring) the rest of the system. It also allows us to obtain information about a player independently of the other players. The partial trace
$\Tr_{-i}(\mathbf{A})$ is an operator and its dimension is equal to the dimension of the subsystem $i$.
One can separate the trace of the operator $\mathbf{A}$ as $\Tr(\mathbf{A}) = \Tr_i(\Tr_{-i}(\mathbf{A}))$, which means that we first perform the partial trace on every subsystem except $i$ and then take the trace over the subsystem $i$.
Next, we need to compute the multiplication of two quantum operators.
If $\mathbf{A}$ and $\mathbf{B}$ be quantum operators for a game of $N$ players, and $\mathbf{C}^{(i)}$ be a quantum operator that acts on the subsystem of the player $i$. Then, the products $\mathbf{A}  \mathbf{B}$ and $\mathbf{A} \mathbf{C}^{(i)}$ can be expressed as
\begin{subequations}
	\label{eq:mulp}
	\begin{align}
		\label{eq:tensor-proda}
		\mathbf{A} \mathbf{B} &= \sum_{\mathbf{j}, \mathbf{j'}} \Big( 
		\sum_{\mathbf{k}} A_{\mathbf{j},\mathbf{k}}B_{\mathbf{k},\mathbf{j'}}
		\Big)
		\ket{\mathbf{j}}\bra{\mathbf{j'}}, \\
		\mathbf{A} \mathbf{C}^{(i)} &= \sum_{\mathbf{j}, \mathbf{j'}}
		\Big( 
		\sum_{k_i} A_{\mathbf{j},(k_i;\mathbf{j'}_{-i})}C_{k_i,j'_i}
		\Big)
		\ket{\mathbf{j}}\bra{\mathbf{j'}}.
		\label{eq:tensor-prodb}
	\end{align}
\end{subequations}

When dealing with quantum operators, computing their traces, performing direct multiplication, and calculating trace values can be computationally expensive and time-consuming. 
As a result, we compute the reward function analytically to avoid the computational cost associated with these operations.
Using Eqs.~(\ref{eq:trace_all}) and (\ref{eq:mulp}), the reward function in Eq.~(\ref{eq:reward-function-density}) can be simplified as 
\begin{align}
	\label{eq:reward_expand}
	r_i(\rho^{(\text{s})}) = 
	\sum\limits_{\mathbf{j},\mathbf{j}'}
	R_{\mathbf{j},\mathbf{j}'}^{(i)}
	\rho_{j'_1,j_1}^{(1)} \rho_{j'_2,j_2}^{(2)} \ldots  \rho_{j'_N,j_N}^{(N)}.
\end{align}

This equation reduces the time complexity of the algorithm from $O\left(\prod_i d_i^3\right)$ to $O\left((N+1) \prod_i d_i^2\right)$. When the reward operator is sparse, the space complexity decreases from $O\left(\prod_i d_i^2\right)$ to $O\left(\sum_i d_i^2 + \sum_i k(\mathbf{R}^{(i)})\right)$, where $k(\mathbf{R}^{(i)})$ is  the number of non-zero entries in $\mathbf{R}^{(i)}$.

To develop the learning algorithm, we need to calculate the gradient of the real-valued function with respect to the complex parameters. For this purpose, we can use the following theorem.
\begin{theorem}[Wirtinger's calculus for a real-valued function, see Refs.~\cite{petersen2008matrix,haykin2002adaptive}]
	\label{theorem:Wirtinger-gradient}
	Let $f: \mathbb{C} \to \mathbb{R}$ be a real-valued function and $z$ be a complex variable. Then, the gradient of $f$ with respect to $z$, denoted by ${\nabla_z} f(z, \bar z)$, is given as
	\begin{align}
		{\nabla_z} f(z, \bar z) = 2 \frac{\partial f(z, \bar z)}{\partial \bar z}.
	\end{align}
\end{theorem}
\begin{proof}
	For a detailed derivation, see Refs.~\cite{petersen2008matrix, haykin2002adaptive}. The factor of 2 in the gradient expression arises from Wirtinger's calculus for complex variables.
\end{proof}
\begin{assumption}
	\label{assum:zzbar}
	Since we work with a real value reward function that contains complex parameters, we employ Wirtinger derivatives. In this approach, we treat $z$ and $\bar z$ as independent variables.
\end{assumption}

Using Assumption~\ref{assum:zzbar} and Eq.~(\ref{eq:reward_expand}), we calculate the derivative of the reward function with respect to the strategy used by player $i$ as  
\begin{align}
	\label{eq:partialr}
	\frac{\partial r_i(\rho^{(\text{s})})}{\partial \bar {\rho}^{(i)}} &= 
	\Tr_{-i}
	\left(
	\mathbf{R}_i \rho^{(\text{s})}_{-i}
	\right).
\end{align}

\begin{definition}
	\label{dif:indi-reward-grad}
	The effective Hamiltonian of a EGW quantum game for player $i$ is defined as
	\begin{align}
		\mathbf{H}^{(i)} (\rho^{\text{(s)}}_{-i})=-  \frac{1}{2} {\nabla _{\rho^{(i)}} } r_i(\rho^{\text{(s)}}).
	\end{align}
\end{definition}
This definition connects the gradient of the reward function to the effective Hamiltonian of the subsystem of the player $i$. This relation calculates the effect of the strategies of other players on the subsystem $i$ in the absence of the player strategy $i$. 
We refer to it as the effective Hamiltonian, since this Hamiltonian represents only one subsystem. The dimensions of the effective Hamiltonian $\mathbf{H}^{(i)}$ and the density matrix $\rho^{(i)}$ are the same.
Using Theorem~\ref{theorem:Wirtinger-gradient}, Eq.~(\ref{eq:partialr}), and Definition~\ref{dif:indi-reward-grad}, the reward function $r_i(\rho^{\text{(s)}})$ can be written as
\begin{align}
	\label{eq:reward_energy_trace}
	r_i  (\rho^{\text{(s)}}) = -
	\Tr_i
	\left(
	\rho^{(i)}
	\mathbf{H}^{(i)}  (\rho^{\text{(s)}}_{-i})
	\right)
\end{align}
In quantum mechanics, $\Tr(\rho \mathbf{H})$ is generally referred to the average energy of the system. In the above theorem, $\Tr_i (\rho^{(i)} \mathbf{H}^{(i)} (\rho^{\text{(s)}}_{-i}))$ can be interpreted as the average energy of the subsystem $i$. Each player aims to maximize the reward or minimize the energy of their own subsystem at temperature $T$ during the learning process.
It is known that EGW quantum games always have at least one Nash equilibrium point~\cite{facchinei2003finite,scutari2010convex,lotidis2023learning}. 
Using Eq.~(\ref{eq:reward_energy_trace}), the Nash equilibrium condition can be expressed as
\begin{align}
	\Tr_i
	\left(
	\mathbf{H}^{(i)}  (\rho^{\text{(s)}*}_{-i})
	\left(
	\rho^{(i)*} - \rho^{(i)} 
	\right)
	\right) \le 0 \quad \forall i \in \mathcal{N}. 
\end{align}
This relation aims to minimize the energy and maximize the reward for all subsystems.

\section{Learning in the EGW quantum game theory}
\label{sec:learning}
In classical game theory, players are generally restricted to using a pure strategy at each decision iteration~\cite{mertikopoulos2016learning}. Even in the case of a classical mixed strategy, players only use one strategy per iteration. However, in the EGW quantum game, each player can employ a mixed quantum strategy at each iteration. The reason is that the density matrix can be expressed as a linear combination of multiple pure quantum strategies. Moreover, even in a pure quantum strategy, the superposition principle allows the strategy to be a combination of multiple actions, each associated with complex probability amplitudes in quantum states.
To find an optimal strategy profile where no player can improve their reward by unilaterally changing their strategy, we propose an optimization technique that is capable of identifying this strategy profile.

In our learning algorithm, we use the von Neumann entropy
to regularize each player's strategy update.
The Von Neumann entropy provides a measure of quantum uncertainty or mixedness in quantum states. 
Higher entropy reflects more uncertainty or mixedness, which encourages the exploration of alternative strategies.
For player $i$, we hold the other players' strategies
$\rho_{-i}^{(\mathrm{s})}$ fixed and consider the following
optimization problem:
\begin{align}
	\underset{\sigma^{(i)}}{\operatorname{minimize}}
	\quad &
	\Tr\!\left(
	\mathbf{H}^{(i)}(\rho_{-i}^{(\mathrm{s})})\sigma^{(i)}
	\right)
	-T S(\sigma^{(i)})
	\nonumber\\
	\text{subject to}\quad &
	\sigma^{(i)}\succeq 0,
	\qquad \Tr(\sigma^{(i)})=1,
	\label{eq:regularized-response}
\end{align}
where $\sigma^{(i)}\in\mathbb{C}^{d_i\times d_i}$,
$T>0$, and
$S(\sigma^{(i)})=-\Tr(\sigma^{(i)}\log\sigma^{(i)})$
is the von Neumann entropy.
The constraints ensure that $\sigma^{(i)}$ is a valid
density matrix.
The first term is the player's expected effective energy,
while the entropy term favors mixed strategies.
The temperature $T$ controls the balance between energy
minimization and entropy maximization.
The following theorem gives the solution of this problem.
\begin{theorem}[Entropy-regularized BR]
	\label{theorem:solutiontypeA}
	Given a quantum game
	$\mathcal{Q}_{\mathrm{EGW}}
	=\langle\mathcal{N},\mathcal{H},\mathcal{S},r\rangle$,
	fix the opponents' strategies
	$\rho_{-i}^{(\mathrm{s})}$ and a temperature $T>0$.
	The unique solution of
	Eq.~\eqref{eq:regularized-response}, which defines
	the entropy-regularized best response, is
	\begin{align}
		\rho_T^{(i)}
		=
		\frac{e^{-\mathbf{H}^{(i)}(\rho_{-i}^{(\mathrm{s})})/T}}
		{\Tr\!\left(	e^{-\mathbf{H}^{(i)}(\rho_{-i}^{(\mathrm{s})})/T}\right)}.
		\label{eq:MMWUA}
	\end{align}
\end{theorem}
\textit{proof:}	See the proof in Appendix~\ref{proof:solutiontypeA}.

The density matrix in Eq.~\eqref{eq:MMWUA} is positive
definite and has unit trace, satisfying the density-matrix
conditions in Eq.~\eqref{eq:dm-conditions}.
For a fixed temperature, requiring every player's strategy
to satisfy Eq.~\eqref{eq:MMWUA} defines a coupled system
of fixed-point equations.

We use these entropy-regularized responses to construct
an iterative method.
Starting from a positive initial temperature, we gradually
decrease the temperature after each iteration.
A larger temperature gives greater weight to entropy,
whereas decreasing the temperature places increasing
emphasis on minimizing the expected effective energy.

\textbf{Matrix Exponential Fixed-Point Iteration with Annealing (MEFPIA) algorithm:}
The MEFPIA algorithm is given by the following iteration rule:
\begin{align}
	\rho^{(i)[k+1]} = \mathbf{M}^{(i)[k]}/\Tr(\mathbf{M}^{(i)[k]}) \quad \forall i \in \mathcal{N},
\end{align}
where $\mathbf{M}^{(i)[k]} = \exp\left(-\mathbf{H}^{(i)}(\rho_{-i}^{(\mathrm{s})[k]}) / T_k \right)$,  
and $T_k = T_0 F(k)$ starts at an initial temperature $T_0$ and gradually decreases to zero over time according to the cooling function $F(k)$, with the initial condition $F(0) = 1$.

\begin{algorithm}[t]
	\caption{\label{alg:al1}
		MEFPIA algorithm for EGW quantum game}
	\begin{algorithmic}[1]
		\Require $T_0$: initial temperature, $\alpha$: cooling rate, $\epsilon$: convergence tolerance
		\State Initialize $\rho^{(i)[0]}$ for each player
		\Repeat{($k=0,1,\cdots$: each step)}
		\For{every player $i \in \mathcal{N}$}
		\State $\mathbf{M}^{(i)[k]} \gets \exp(-\mathbf{H}^{(i)} (\rho^{(s)[k]}_{-i})/T_k)$
		\State $\rho^{(i)[k+1]} \gets \mathbf{M}^{(i)[k]}/\Tr(\mathbf{M}^{(i)[k]})$
		\EndFor
		% \State $T_{k+1} \gets \alpha T_k$
		\State Update temperature $T_{k+1}$
		\Until{$\left|\left|\rho^{(i)[k+1]}-\rho^{(i)[k]}\right|\right|\le \epsilon$ for all players}
	\end{algorithmic}
\end{algorithm}
Algorithm~\ref{alg:al1} presents the steps of MEFPIA.
At a fixed positive temperature, a joint fixed point
of Eq.~\eqref{eq:MMWUA} consists of mutually
entropy-regularized best responses.

If the strategy iterates converge to a profile
$\rho^{(\mathrm{s})*}$ while $T_k\to0$, continuity of
the effective Hamiltonians and the vanishing entropy
regularization imply
\begin{align}
	\rho^{(i)*}
	\in
	\underset{\substack{\sigma^{(i)}\succeq0\\
			\Tr(\sigma^{(i)})=1}}{\arg\min}
	\Tr\!\left(
	\mathbf{H}^{(i)}(\rho_{-i}^{(\mathrm{s})*})\sigma^{(i)}
	\right),
	\qquad \forall i\in\mathcal{N}.
\end{align}
Each limiting strategy is therefore a best response to
the limiting strategies of the other players, so the
resulting profile is a Nash equilibrium.

We refer to the gradual decrease of the regularization parameter $T$ as temperature annealing.
A higher temperature gives greater weight to entropy, while a lower temperature places greater emphasis on payoff maximization.
The cooling schedule controls how quickly this balance changes during the iterations.
There are several common cooling functions as described in the literature~\cite{hajek1988cooling,kirkpatrick1983optimization,nourani1998comparison,karabin2020simulated,aarts1989simulated}. Examples are as follows.
\begin{enumerate}
	\item Geometric cooling: $F(k) = \alpha^{k}$, where $\alpha \in (0,1)$.
	\item Logarithmic cooling: $F(k) = {\log \alpha}/{\log(k + \alpha)}$, where $\alpha > 1$.
	\item  Linear cooling over a finite iteration horizon: 	$F(k)=1-\alpha k$, where $0\leq k\leq k_{\max}$ and
	$0<\alpha<1/k_{\max}$, ensuring $T_k>0$ throughout the computation.
	\item Exponential decay: $F(k) = \exp(-\alpha k)$, where $\alpha > 0$.
	\item Adaptive cooling: $F(k) = (1 - \alpha_k) F(k-1) $, where $\alpha_k \in (0,1)$ is adaptively updated during iterations based on optimization feedback.
\end{enumerate}
The parameter $\alpha$, also called the cooling rate, should be fine-tuned to control the cooling speed.
For example, the exponential cooling function may cool too quickly if $\alpha$ is too large. However, the logarithmic cooling function may result in slower cooling, which could lead to slower convergence compared to geometric cooling. 

For a fixed effective Hamiltonian, the Gibbs response concentrates on its minimum-eigenvalue eigenspace as
$T\to0$. If the smallest eigenvalue is nondegenerate, the limiting response is pure and if it is degenerate,
the limiting response can be mixed.
This result is in contradiction with the findings of Lotidis \etal~\cite{lotidis2024payoff}, where all strategies were shown to converge to pure density matrices.
Lotidis \etal used projective outcomes in their experiment, where the outcomes are orthogonal to each other. 
In Experiment 2 (Section~\ref{sec:experiment2}), we used a general POVM set with non-orthogonal outcomes. Our results showed that in this setting the optimal strategy is no longer a pure state.

\section{experimental analysis}
\label{sec:expriments}
In this experimental section, we study and compare the convergence of MMWU and MEFPIA in different setups of the EGW quantum game.
To improve readability, we aggregate the outcome rewards and utilize vector notation
$\boldsymbol{r}_i = (R_i(\omega_{0}), R_i(\omega_{1}), \ldots,R_i(\omega_{m}))$ in Eq.~\eqref{eq:reward-op-main}.
The $R_i(\omega_{j})$ is the outcome reward corresponding to the $j$-th measurement outcome of player $i$, and $m$ is the total number of outcomes. 
In these experiments, we use the geometric cooling function $F(k) = \alpha^{k}$ during the cooling process. Since small values of $\alpha$ can cause the system to cool too quickly, we choose $\alpha$ close to one, but slightly less than one during the learning steps.

To study the convergence, we use the relative error, defined as
\begin{align}
	\varepsilon_{\text{rel}} = \frac{| r - r_0 |}{| r_0 |},
\end{align}
where $r$ is the expected reward at each iteration, and $r_0$ is the expected reward at the fixed-point.
In these experiments, we used 100 random initial strategies for all players and averaged the relative error over the training iterations. We also plot the variance of the relative error in lighter colors in the figures. To produce meaningful logarithmic scaling, we only consider the upper bound of the variance, which gives a more apparent and understandable appearance.

In the following, we study two different experiments. 
In experiment~1, we use orthogonal outcomes (projective outcomes) and apply the reward of the classical Prisoner's Dilemma in a quantum format. We divide experiment~1 into three mini-experiments, each using a different set of outcomes. In the first mini-experiment, we use simple separable outcomes that map to the classical Prisoner's Dilemma. In the second mini-experiment, we use more complex separable outcomes, and in the last mini-experiment, we use entangled outcomes.
In experiment~2, we use a larger strategy space and a larger set of outcomes in a three-player game.

In all experiments, we compare the performance of MMWU and MEFPIA and show that the MEFPIA algorithm converges faster with greater accuracy.
In the experiments, $\gamma$ denotes the learning rate of MMWU~\cite{lotidis2024payoff}.
We also investigate the effect of the cooling rate on the convergence speed of MEFPIA.

\subsection{Experiment 1}
In this experiment, we use an extension version of the classical Prisoner's Dilemma to the EGW quantum game framework. In the classical version of the Prisoner's Dilemma, two players can either cooperate or defect (C or D). They do not know the strategy of the other player. If both players cooperate to each other, they each receive a payoff of 3. If both defect, they each receive a payoff of 1. If one defects while the other cooperates, the defector receives a payoff of 5, and the cooperator receives 0. 
The Nash equilibrium in this game occurs when both players defect, as this is the best response to the opponent's strategy. Although mutual cooperation, which corresponds to a payoff of 3 for both players, would lead to a better outcome for both, the incentive to defect makes mutual defection, corresponding to a payoff of 1 for each player, the dominant strategy in this game. 
The incentive to defect makes the equilibrium at mutual defection with payoff 1 rather than mutual cooperation payoff 3, even though the latter would be better for both players.
In classical scenarios, players can use a mixed strategy, which means that they choose randomly between cooperation and defection based on predefined probabilities.

The Prisoner's Dilemma in a EGW quantum game is defined as  $\mathcal{Q}_{\text{EGW}}^{\text{PD}} = \langle \mathcal{N}=\{1,2\}, \mathcal{H}=\mathcal{H}^{(1)} \otimes \mathcal{H}^{(2)}, \mathcal{S}, r \rangle$ with  $d_i=2$. The outcome reward vector $r_i$ corresponds to the set of outcomes and the reward of each outcome.

To provide an overview of this quantum game and establish a simple mapping between the classical and quantum versions of the Prisoner's Dilemma, we conduct three mini-experiments with different outcome sets in this study. In these mini-experiments, we use an outcome reward vector for each player, similar to the classical Prisoner's Dilemma, as follows
\begin{align}
	\label{eq:reward1-exp}
	\mathbf{r}_1 = (3,0,5,1), \quad \mathbf{r}_2 = (3,5,0,1),
\end{align}	
which are related to the outcome set $\Omega = \{ \omega_1, \cdots, \omega_4\}$.

\textbf{Mini-experiment 1:}
In this mini-experiment, we look at a one-to-one mapping between classical and quantum Prisoner's Dilemma.
In comparison with the classical Prisoner's Dilemma with mixed strategies, $\ket{0}$ represents cooperation, and $\ket{1}$ represents defection.
In the density matrix formalism, we represent the cooperative strategy of the player $i$ as ${\rho^{(i)} = \ket{0}\bra{0}}$, and the defect strategy of player $i$ as ${\rho^{(i)} = \ket{1}\bra{1}}$. The strategy profile in which the defect is chosen by both players is represented by $\rho^{(\text{s})} = \ket{11}\bra{11}$.

The outcomes and their classical counterparts are presented as follows
\begin{align}
	\omega_1 \equiv CC , &\quad
	\omega_2 \equiv CD , \nonumber \\
	\omega_3 \equiv DC , &\quad
	\omega_4 \equiv DD ,
\end{align}
where $C$ denotes cooperation and $D$ denotes defection. $CD$ means that the first player cooperates, while the second player defects.
The Nash equilibrium corresponds to a payoff of 1 for each player when both players use the defect strategy. 
For the above outcomes, the outcome operators $\mathbf{P}_\omega$ are defined as
\begin{align}
	\mathbf{P}_{\omega_{1}} = \ket{00}\bra{00}, &\quad
	\mathbf{P}_{\omega_{2}} = \ket{01}\bra{01}, \nonumber \\
	\mathbf{P}_{\omega_{3}} = \ket{10}\bra{10}, &\quad
	\mathbf{P}_{\omega_{4}} = \ket{11}\bra{11}.
\end{align}
Our experiments show that for both the MMWU and MEFPIA algorithms, the strategies and payoffs converge to this equilibrium point, which is in agreement with the classical Prisoner's Dilemma.

\begin{figure}[!tbp]
	\includegraphics[width=0.99\linewidth]{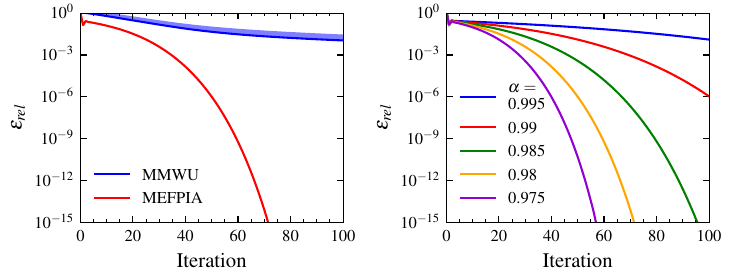}
	\caption{Relative error for mini-experiment 1 (left)	MMWU ($\gamma=0.05$) and MEFPIA ($T_0=1, \alpha = 0.98$) (right) MEFPIA with $T_0=1$ with different cooling rates.  }
	\label{fig:PrisonersDilemma}
\end{figure}
In Fig.~\ref{fig:PrisonersDilemma}-(left), we compare the convergence of MEFPIA and MMWU to the optimal strategies. The MEFPIA converges to the optimal solution much faster than the MMWU. The relative error for MMWU is greater than $10^{-3}$, while the relative error for our algorithm, MEFPIA, can reach machine precision ($10^{-16}$) in fewer than 100 iterations.
Both algorithms converge to the same optimal strategy profile and receive the same payoff at the convergence point, which is equivalent to both players defecting in the classical Prisoner's Dilemma.
In Fig.~\ref{fig:PrisonersDilemma}-(right), we study the effect of the cooling rate $\alpha \in (0,1)$ on MEFPIA. A higher cooling rate leads to slower cooling, causing the algorithm to converge after more iterations, whereas a lower cooling rate causes faster cooling. In many circumstances, the system should have enough time to explore the strategy space, as cooling too quickly may prevent sufficient exploration. Therefore, we tune the cooling rate to an optimal value, avoiding both extremes of being too small or too large.

\textbf{Mini-experiment 2:}
In this mini-experiment, we use the separate outcome set. Consider that each private set of outcomes contains two members as $\Omega^{(i)} = \{ \omega_1^{(i)}, \omega_2^{(i)} \}$, with the corresponding outcome operators $\mathbf{P}_\omega$ given by
\begin{align}
	\mathbf{P}_{\omega_{1}^{(i)}} = \frac{1}{2}\big(\ket{0}+i\ket{1}\big)\big(\bra{0}-i\bra{1}\big), \nonumber \\
	\mathbf{P}_{\omega_{2}^{(i)}} = \frac{1}{2}\big(\ket{0}-i\ket{1}\big)\big(\bra{0}+i\bra{1}\big).
\end{align}
These outcome operators correspond to each player $i$, and the private outcome set is the same for both players.
As a result, there are four outcomes in the separable outcome set.

\begin{figure}
	\centering
	\includegraphics[width=0.99\linewidth]{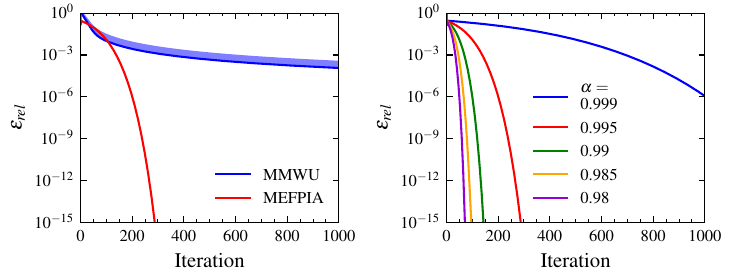}
	\caption{Relative error for mini-experiment 2.
		(left)	MMWU ($\gamma=0.05$) and MEFPIA ($T_0=1, \alpha = 0.995$) (right) MEFPIA with $T_0=1$ with different cooling rates.}
	\label{fig:PrisonersDilemma-sep-ind}
\end{figure}
Fig.~\ref{fig:PrisonersDilemma-sep-ind}-(left) shows the performance of MEFPIA and MMWU. The MEFPIA converges to the optimal solution significantly faster than the MMWU. 
In Fig.~\ref{fig:PrisonersDilemma-sep-ind}-(right), we study the convergence of the MEFPIA algorithm for different values of the cooling rate. 

In Fig.~\ref{fig:PrisonersDilemma-sep-bloch}, we show the players' strategies (density matrices) at each iteration on the Bloch sphere.
To convert a $2 \times 2$ density matrix $\rho$ to the Bloch sphere representation~\cite{nielsen2010quantum}.
The pure states correspond to points on the surface, and the mixed states lie inside the sphere.

\begin{figure}[t]
	\centering
	\includegraphics[width=0.99\linewidth]{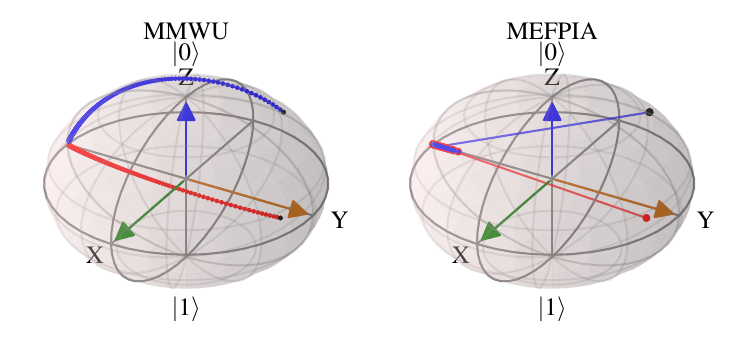}
	\caption{Bloch sphere for Prisoner's Dilemma quantum game with non-entangled outcome. Player 1 is represented by blue color, and player 2 by red color.
		(left)	MMWU ($\gamma=0.05$) and (right) MEFPIA ($T_0=1, \alpha = 0.99$)}
	\label{fig:PrisonersDilemma-sep-bloch}
\end{figure}
In Fig.~\ref{fig:PrisonersDilemma-sep-bloch}, we start training at the black points for both players, using both MEFPIA and MMWU. 
During training, the strategies converge to the optimal strategy. Fig.~\ref{fig:PrisonersDilemma-sep-bloch}-(left) corresponds to the MMWU, which smoothly converges to the convergence point. On the other hand, in Fig.~\ref{fig:PrisonersDilemma-sep-bloch}-(right), the MEFPIA shows a jump in the initial steps and then focuses on finding the solution around the optimal point. This behavior of both algorithms is related to the entropy used in their development.

\textbf{Mini-experiment 3:}
In this mini-experiment, we use a set of outcomes that includes entangled states, which generates the entangled reward operator. Consider the outcome operators as 
\begin{align}
	\label{eq:bell-exp}
	\mathbf{P}_{\omega_{1}} &= \frac{1}{2} \big(\ket{00}+i\ket{11}\big)\big(\bra{00}-i\bra{11}\big), \nonumber \\
	\mathbf{P}_{\omega_{2}} &= \frac{1}{2} \big(\ket{01}-i\ket{10}\big)\big(\bra{01}+i\bra{10}\big), \nonumber \\ 
	\mathbf{P}_{\omega_{3}} &= \frac{1}{2} \big(\ket{10}-i\ket{01}\big)\big(\bra{10}+i\bra{01}\big), \nonumber \\ 
	\mathbf{P}_{\omega_{4}} &= \frac{1}{2} \big(\ket{11}+i\ket{00}\big)\big(\bra{11}-i\bra{00}\big). 
\end{align}
Since there is entanglement between the outcome subsystems and the strategy profile is separable, no strategy profile can capture each $\mathbf{P}_{\omega_{j}}$ with a probability of $1$. As a result, it cannot cover the entire reward space.
The maximum probability that a strategy profile can achieve for these entangled outcomes is $0.5$. 
In this experiment, both MEFPIA and MMWU converge to the same strategy profile with a payoff $2.25$.

\begin{figure}[b]
	\centering
	\includegraphics[width=0.99\linewidth]{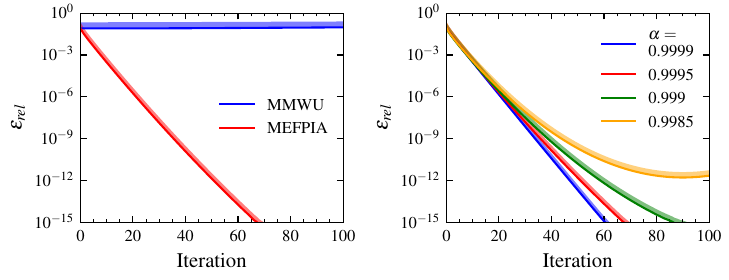}
	\caption{Relative error for mini-experiment 3 with entangled outcomes. (left)	MMWU ($\gamma=0.01$) and MEFPIA ($T_0=3, \alpha = 0.9995$) (right) MEFPIA with $T_0=3$ with different values for cooling rate.}
	\label{fig:PrisonersDilemma-ent}
\end{figure}
In Fig.~\ref{fig:PrisonersDilemma-ent}-(left), the convergence of both MEFPIA and MMWU is studied. As in the other experiments, MEFPIA performs better than MMWU in terms of convergence. In Fig.~\ref{fig:PrisonersDilemma-ent}-(right), the convergence of MEFPIA is studied for different values of the cooling rate. Due to the nature of the entangled outcomes, the cooling rates differ significantly from those used in the previous mini-experiments.

In Experiment 1, all outcomes are orthogonal to each other (projective outcomes), and the optimal strategies converged to the pure density matrices with $\Tr(\rho^{(i)2}) = 1$.

\subsection{Experiment 2:}
\label{sec:experiment2}
In the previous experiment, the number and size of the outcome spaces were small, and all outcomes were orthogonal to each other (projective outcomes). However, a general set of outcomes in a POVM does not require orthogonal outcomes. To study the performance of MEFPIA and MMWU in a more complex scenario, in this experiment, we use non-orthogonal and separable outcomes with 3 players. Additionally, the dimension of the subsystem is $d_i = 5$, and the number of private outcomes is $m_i = 3$ for each player $i$. The final set of outcomes for the entire system consists of $3^3 = 27$ outcomes, with dimension $5^3 = 125$. The detailed information regarding the set of outcomes and the reward vector is discussed in the Appendix~\ref{append:add-exprim}.

\begin{figure}
	\centering
	\includegraphics[width=0.99\linewidth]{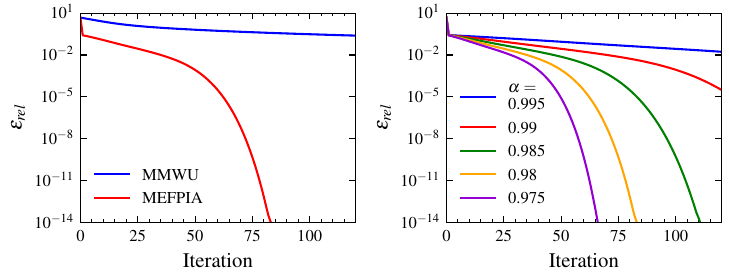}
	\caption{Relative error for Experiment 2 with separable outcomes and large system. 
		(left)	MMWU ($\gamma=0.01$) and MEFPIA ($T_0=1, \alpha = 0.98$) (right) MEFPIA with $T_0=1$ with different cooling rates.}
	\label{fig:experiment2}
\end{figure}
In Fig.~\ref{fig:experiment2}-(left), we compare the convergence of MMWU and MEFPIA. Fig.~\ref{fig:experiment2}-(right) shows the effect of the cooling rate on the performance of MEFPIA. 
In this experiment, the final strategies are mixed density matrices with $\Tr[(\rho^{(i)})^2]\approx 1/2$.

\textbf{Discussion on Overflow Limits at Low Temperatures:}
For a large system, such as the model used in this experiment, the calculation of 
$ \mathbf{M}^{(i)[k]}$ 
and $\rho^{(i)[k+1]} = \mathbf{M}^{(i)[k]}/\Tr(\mathbf{M}^{(i)[k]})$
at very low temperatures becomes challenging, since the values of $\mathbf{M}^{(i)[k]}$ can exceed the overflow limit of the code (Algorithm~\ref{alg:al1}). In these scenarios, we first calculate the eigenvalues and eigenvectors of $\mathbf{A} =-\mathbf{H}^{(i)} (\rho^{(s)[k]}_{-i})/T_k$ before calculating the exponential. Consider $\Lambda$ as a diagonal matrix of eigenvalues $(\lambda_1, \lambda_2, \dots, \lambda_{d_i})$, and $\mathbf{V}$ as the matrix of eigenvectors of $\mathbf{A} = \mathbf{V} \Lambda \mathbf{V}^\dagger$.
The exponential of $\mathbf{A}$ can be calculated as
$\exp(\mathbf{A}) = \mathbf{V}\exp(\Lambda) \mathbf{V}^\dagger$. 
Let $\lambda_{\max}=\max_{1\leq \ell\leq d_i}\lambda_\ell$ 	be the largest eigenvalue of $\mathbf{A}$. 
Then
\begin{align}
	e^{\mathbf{A}} = \mathbf{V}e^{\Lambda} \mathbf{V}^\dagger = e^{\lambda_{\text{max}}} \mathbf{V}e^{\Lambda-\lambda_{\text{max}}\mathbf{I}} \mathbf{V}^\dagger.
\end{align}
The term $e^{\Lambda - \lambda_{\text{max}} \mathbf{I}}$ can be easily calculated and remains within the overflow limit. Similarly, the trace of the exponential operator can be written as
\begin{align}
	\Tr(\exp(\mathbf{A})) = e^{\lambda_{\text{max}}} \, \Tr\left( e^{\Lambda - \lambda_{\text{max}} \mathbf{I}} \right).
\end{align}
Since the factor $e^{\lambda_{\text{max}}}$ cancels both in the numerator and in the denominator, we can easily calculate the next step of the algorithm~\ref{alg:al1} for $\rho^{(i)[k+1]}$ without encountering any problems. It is worth noting that calculating eigenvalues and eigenvectors can be time consuming. Therefore, we use this method only at very low temperatures, when necessary.

\textbf{Reproducibility Statement:}The experimental setup, implementation steps, algorithmic
procedures, and parameter choices are described in the main
text and appendices to support reproducibility.
ChatGPT and QuillBot were used to improve the language
and readability of this paper.

\section{Conclusions}
\label{sec:conclud}
In this paper, we simplified payoff and gradient calculations in EGW quantum games using tensor-contraction
expressions. 
These expressions avoid explicitly constructing the joint density matrix and multiplying it by the reward
operators, reducing the computational and memory costs associated with direct evaluation using full matrix products.
We then proposed the MEFPIA algorithm for equilibrium search, combining entropy-regularized best responses with
a gradually decreasing temperature. 
We compared MEFPIA with MMWU under different measurement configurations, including separable and entangled
measurement outcomes. 
For the tested instances and parameter settings, both algorithms approached the same strategy profiles and payoffs, while MEFPIA achieved lower relative payoff error in fewer iterations.
In future work, we will investigate properties of MEFPIA and evaluate its
performance in larger quantum games with more general measurement structures. We also plan to investigate
practical applications of quantum games in multi-agent decision-making.

\appendix
\section{}

\subsection{Proofs}

\label{proof:solutiontypeA}
\textbf{Proof of Theorem~\ref{theorem:solutiontypeA}:}
For the player $i$, fix the other player's strategies
$\rho_{-i}^{(\mathrm{s})}$, and a temperature $T>0$.
For brevity, write
$\mathbf{H}=\mathbf{H}^{(i)}(\rho_{-i}^{(\mathrm{s})})$
and $\sigma=\sigma^{(i)}$.
The optimization problem in
Eq.~\eqref{eq:regularized-response} minimizes
\begin{align}
	F_T(\sigma)
	=
	\Tr(\mathbf{H}\sigma)
	+T\Tr(\sigma\log\sigma),
\end{align}
subject to $\sigma\succeq0$ and $\Tr(\sigma)=1$,
with the condition $0\log0=0$.
Introducing a Lagrange multiplier $\mu\in\mathbb{R}$
for the unit-trace constraint gives
\begin{align}
	\mathcal{L}(\sigma,\mu)
	&=
	\Tr(\mathbf{H}\sigma)
	+T\Tr(\sigma\log\sigma)
	\nonumber\\
	&\quad+\mu\big(\Tr(\sigma)-1\big).
\end{align}
We first find a positive-definite candidate
and then verify its global optimality over all feasible
density matrices.
For $\sigma\succ0$, the first variation in a Hermitian
perturbations of $\delta\sigma$ is
\begin{align}
	\delta\mathcal{L}
	=
	\Tr\!\left[
	\left(
	\mathbf{H}
	+T(\log\sigma+\mathbf{I}^{(i)})
	+\mu\mathbf{I}^{(i)}
	\right)\delta\sigma
	\right].
\end{align}
Setting this variation to zero for every Hermitian
$\delta\sigma$ yields
\begin{align}
	\mathbf{H}
	+T(\log\sigma+\mathbf{I}^{(i)})
	+\mu\mathbf{I}^{(i)}
	=0.
\end{align}
Therefore,
\begin{align}
	\sigma
	&=
	\exp\!\left(
	-\frac{\mathbf{H}}{T}
	-\left(1+\frac{\mu}{T}\right)\mathbf{I}^{(i)}
	\right)
	\nonumber\\
	&=
	e^{-1-\mu/T}\exp(-\mathbf{H}/T).
\end{align}
The second equality follows from the
Baker-Campbell-Hausdorff formula, since
$[\mathbf{H},\mathbf{I}^{(i)}]=0$, together with
$\exp(c\mathbf{I}^{(i)})=e^c\mathbf{I}^{(i)}$
for any scalar $c$.
The unit-trace constraint gives
\begin{align}
	e^{-1-\mu/T}
	=
	\frac{1}{Z_T},
	\qquad
	Z_T=\Tr\!\left(\exp(-\mathbf{H}/T)\right).
\end{align}
Hence, the stationary candidate is
\begin{align}
	\rho_T^{(i)}
	=
	\frac{\exp(-\mathbf{H}/T)}{Z_T}.
\end{align}
Since $\mathbf{H}$ is Hermitian and $T>0$,
$\rho_T^{(i)}$ is positive definite and has unit trace.
This result is an application of the Gibbs variational
principle, which characterizes the Gibbs state as the
minimizer of the entropy-regularized energy
functional~\cite{shi2020variational}.

To establish global optimality and uniqueness, 
consider any feasible density matrix $\sigma$ and define the
Umegaki quantum relative entropy~\cite{umegaki1962conditional}
 with respect to $\rho_T^{(i)}$ as
\begin{align}
	D(\sigma\Vert\rho_T^{(i)})
	=
	\Tr(\sigma\log\sigma)
	-\Tr(\sigma\log\rho_T^{(i)}).
\end{align}
This quantity is finite because $\rho_T^{(i)}$ is
positive definite.
Using
\begin{align}
	\log\rho_T^{(i)}
	=
	-\frac{\mathbf{H}}{T}
	-(\log Z_T)\mathbf{I}^{(i)},
\end{align}
we obtain~\cite{brandao2015second}:
\begin{align}
	T D(\sigma\Vert\rho_T^{(i)})
	&=
	T\Tr(\sigma\log\sigma)
	+\Tr(\mathbf{H}\sigma)
	+T\log Z_T
	\nonumber\\
	&=F_T(\sigma)+T\log Z_T.
\end{align}
Setting $\sigma=\rho_T^{(i)}$ gives
$F_T(\rho_T^{(i)})=-T\log Z_i$.
Therefore,
\begin{align}
	F_T(\sigma)-F_T(\rho_T^{(i)})
	=
	T D(\sigma\Vert\rho_T^{(i)})
	\geq 0.
\end{align}
By the nonnegativity of Umegaki quantum relative entropy
(Klein's inequality)~\cite{watrous2018theory}, equality holds if and only if
$\sigma=\rho_T^{(i)}$.
Thus, $\rho_T^{(i)}$ is the unique minimizer of
Eq.~\eqref{eq:regularized-response}, proving
Eq.~\eqref{eq:MMWUA}.

\subsection{Generating POVM Elements}
\label{append:povm_generate}

In the second experiment, we generate a POVM with $m$
outcomes for a system of dimension $d$.
The outcome set is $\Omega=\{\omega_j\}_{j=0}^{m-1}$.
The corresponding POVM elements must satisfy positivity
and completeness~\cite{watrous2018theory}:
\begin{align}
	\mathbf{P}_{\omega_j}\succeq 0,
	\qquad
	\sum_{j=0}^{m-1}\mathbf{P}_{\omega_j}=\mathbf{I}.
\end{align}
We construct these elements using the following steps:
\begin{enumerate}
	\item Define the orthonormal basis
	$\{\mathbf{v}_k\}_{k=0}^{d-1}$:
	Select $d$ orthonormal vectors in $\mathbb{C}^d$
	and use them as the columns of a unitary matrix
	$\mathbf{U}\in\mathbb{C}^{d\times d}$.
	
	\item Define the diagonal matrices:
	For each outcome $j$, select a vector
	$\mathbf{u}_j=(u_{j_0},\ldots,u_{j_{d-1}})$
	with real, nonnegative entries satisfying
	\begin{align}
		u_{j_k}\geq 0,
		\qquad
		\sum_{j=0}^{m-1}u_{j_k}=1,
		\quad k=0,\ldots,d-1.
	\end{align}
	Set $\mathbf{D}_j=\operatorname{diag}(\mathbf{u}_j)$.
	
	\item Generate the POVM outcome operators:
	Define
	\begin{align}
		\mathbf{P}_{\omega_j}
		=\mathbf{U}\mathbf{D}_j\mathbf{U}^{\dagger},
		\quad j=0,\ldots,m-1.
	\end{align}
\end{enumerate}

Each $\mathbf{D}_j$ is positive semidefinite, so
$\mathbf{P}_{\omega_j}$ is also positive semidefinite.
Moreover, the normalization of the diagonal entries gives
\begin{align}
	\sum_{j=0}^{m-1}\mathbf{P}_{\omega_j}
	&=
	\mathbf{U}
	\left(\sum_{j=0}^{m-1}\mathbf{D}_j\right)
	\mathbf{U}^{\dagger}
	\nonumber\\
	&=\mathbf{U}\mathbf{I}\mathbf{U}^{\dagger}
	=\mathbf{I}.
\end{align}
Thus, the generated operators form a valid POVM.
Since all elements are diagonal in the same basis, this construction generates commuting POVMs.

\subsection{Additional Information on Experimental Analysis}
\label{append:add-exprim}
In experiment~2 in Section~\ref{sec:experiment2}, we generated the non-orthogonal POVM set of outcomes using the method described in Appendix~\ref{append:povm_generate}. 
We considered a 3-player game where each player $i$ has access to a private set of outcomes $\Omega_i = \{\omega_{1}^{(i)},\omega_2^{(i)},\omega_{3}^{(i)}\}$ for a separable reward operator. For simplicity, we assume that the private set of outcomes is the same for all players. 
We used $\mathbf{U}$ and $\{ \mathbf{u}_j \}$ as 
\begin{subequations}
	\begin{align}
		\mathbf{U} &= \begin{pmatrix}
			\frac{1}{\sqrt{3}} & \frac{1}{\sqrt{3}} & \frac{1}{\sqrt{3}}&0&0 \\
			\frac{1}{\sqrt{2}} & -\frac{1}{\sqrt{2}} & 0 &0&0\\
			\frac{1}{\sqrt{6}} & \frac{1}{\sqrt{6}} & -\sqrt{\frac{2}{3}}&0&0\\
			0&0&0&\frac{1}{\sqrt{2}} & \frac{i}{\sqrt{2}}  \\
			0&0&0&\frac{1}{\sqrt{2}} & -\frac{i}{\sqrt{2}} 
		\end{pmatrix}, \\
		\mathbf{D}_0 &=\text{diag}(0.5,0.50,0.1,0.0,0.0),
		\\
		\mathbf{D}_1 &=\text{diag}(0.5,0.25,0.4,0.5,0.5),
		\\
		\mathbf{D}_2 &=\text{diag}(0  ,0.25,0.5,0.5,0.5).
	\end{align}
\end{subequations}
We used Appendix~\ref{append:povm_generate} to generate the private outcome operators for each player and used them for the separable outcome set of the entire system. Then, we set the reward vector as $\mathbf{r}_i = (1, 2, \dots, 27)$ for all players and applied Eq.~\eqref{eq:reward-op-main} to generate the reward operators.

\bibliography{references}
\bibliographystyle{IEEEtran}

\end{document}